\documentclass[10pt,conference]{IEEEtran}

\usepackage{times}
\usepackage[cmex10]{amsmath}
\usepackage{amssymb}
\usepackage{epsfig,verbatim}

\newcommand{\eqnlabel}[1]{\label{eqn:#1}}
\newcommand{\eqnref}[1]{(\ref{eqn:#1})}

\providecommand{\Wc}{{\cal W}}

\newtheorem{MyTheorem}{Theorem}
\newcommand{\thmlabel}[1]{\label{thm:#1}}
\newcommand{\thmref}[1]{\ref{thm:#1}}

\newtheorem{MyLemma}{Lemma}

\newtheorem{MyRemark}{Remark}

\title{Low Latency Communications}
\begin{document}

%%\author{
%%\authorblockN{Ivana Mari\'c}
%%\authorblockA{Stanford U\\
%%Stanford, CA \\
%%ivanam@wsl.stanford.edu} \and
%%\authorblockN{Ron Dabora}
%%\authorblockA{Stanford University\\
%%Stanford, CA \\
%%ron@wsl.stanford.edu} \and
%%\authorblockN{Andrea Goldsmith}
%%\authorblockA{Stanford University  \\
%%Stanford, CA \\
%%andrea@wsl.stanford.edu } }

\author{
\authorblockN{Ivana Mari\'c}
\authorblockA{Aviat Networks\\
Santa Clara, CA \\
ivana.maric@aviatnet.com}
}
%%\authorblockN{Ron Dabora}
%%\authorblockA{Stanford University\\
%%Stanford, CA \\
%%ron@wsl.stanford.edu} \and
%%\authorblockN{Andrea Goldsmith}
%%\authorblockA{Stanford University  \\
%%Stanford, CA \\
%%andrea@wsl.stanford.edu } }

%\author{
%\authorblockN{Ivana Mari\'c}
%\and
%\authorblockN{Andrea J. Goldsmith
%}
%\authorblockA{Stanford University \\
%Stanford, CA\\
%andrea@wsl.stanford.edu} 
%}

\maketitle

\date{}
%%%%%%%%%%%%%%%%%%
\begin{abstract}
Numerous applications demand communication schemes that minimize the transmission delay while achieving a given level of reliability. An extreme case is high-frequency trading whereby saving a fraction of millisecond over a route between Chicago and New York can be a game-changer. While such communications are often carried by fiber, microwave links can reduce transmission delays over large distances due to more direct routes and faster wave propagation. In order to bridge large distances, information is sent over a multihop relay network. 

Motivated by these applications, this papers present an information-theoretic approach to the design of optimal multihop microwave networks that minimizes end-to-end transmission delay. 
To characterize the delay introduced by coding, we derive error exponents achievable in multihop networks.
We formulate and solve an optimization problem that determines optimal selection of amplify-and-forward and decode-and-forward relays.  
We present the optimal solution for several examples of networks. We prove that in high SNR the optimum transmission scheme is for all relays to perform amplify-and-forward.  We then analyze the impact of deploying noisy feedback.  
\end{abstract}
%%%%%%%%%%%%%%%%%%%%%%%%
\section{Introduction}
%%%%%%%%%%%%%%%%%%%%%%%%
%%%%%%%%%%%%%%%%%%%%%%%%%%%%%%
% 1. Low-latency Applications
%%%%%%%%%%%%%%%%%%%%%%%%%%%%%%
Operating close to the channel capacity requires encoding with large codelengths in order to guarantee diminishing probability of error.
In turn, large codelengths introduce decoding delay at the receiver.
If data is sent over a multihop network, this delay can multiply over multiple hops, thereby increasing the end-to-end latency. 
On the other hand, numerous applications, instead of striving to operate at the maximum rate, demand communications with the minimum latency. 
An extreme case is {high-frequency trading} in which profits depend on computer-based algorithmic trades that are made as fast as possible.
In these settings, in order to bridge large distances, information is sent over a multihop relay network (see Fig.~\ref{Fig2}). 
At the same time, saving a fraction of millisecond over, for example, a route between Chicago and New York, can be a game-changer. 
According to \cite{news}, 1 ms of reduced delay translates into \$100 million profit per year. 
On the Chicago-New York route, fiber can deliver data in 6.6 ms. On the other hand, the latest microwave network can deliver data in 4.1 ms \cite{PressRelease}.
The gain in the microwave transmission comes from faster wave propagation in the air when compared to the fiber, and from shorter routes. 
In this paper, we are concerned with such low-latency communications.

%%%%%%%%%%%%%%%%%%%%%%%%%
% 2. Microwave network
%%%%%%%%%%%%%%%%%%%%%%%%%
%We consider a  multihop relay network with one source-destination pair (see Fig.~\ref{Fig2}) with and without feedback. 
End-to-end transmission delay drastically varies with the choice of the cooperative scheme used by a relay. In particular, a relay performing decode-and-forward (DF) will introduce a delay of the order of the size of the block, $n$, that it needs to receive prior to decoding, re-encoding and forwarding a message. In contrast, an amplify-and-forward (AF) relay can forward on per-symbol basis, thus introducing a roughly $n$ times smaller delay. 
%Furthermore, the processing delay introduced at each node  is significantly lower for AF. Specifically, the processing delay for AF is of the order of ns, whereas for DF it is in %the order of  $\mu$s.
However, the simplicity of AF comes at the expense of amplifying and propagating the noise thereby reducing the effective received signal-to-noise ratio with every subsequent AF hop. The reduced SNR reduces the transmission rate and ultimately, after a sequence of AF hops, results in a higher delay when compared to DF. 

%%%%%%%%%%%%%%%%%%
% Common practice
%%%%%%%%%%%%%%%%%%
In microwave low-latency networks, common practice is to perform DF or AF at a node uniquely  on the basis of its received SNR.  A relay with a received SNR above a certain threshold performs AF, otherwise, it performs DF. Typically, this criterion performs DF at a relay after a several AF hops or after one long hop.

%%%%%%%%%%%%%%%%%%%%%%%%%%%%%%%%%%%%%%%%%%%%
\begin{figure}[t]
\centering
\includegraphics[height=.9cm, width=3.5in]{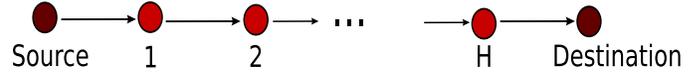}
\caption{Multihop network. Relays are labeled as 1,\ldots,H}
      \label{Fig2}
\end{figure}
%%%%%%%%%%%%%%%%%%%%%%%%%%%%%%%%%%%%%%%%%%%%
%%%%%%%%%%%%%%%%%%
%3. Open questions
%%%%%%%%%%%%%%%%%%
Our goal is to design a multihop microwave network to minimize the end-to-end delay. Towards that goal, we address several practical questions: 1) Is the common practice of assigning AF/DF relays based on received SNR optimal? 2) If not, when should DF relays be used given that they introduce larger delays? 3) Does this selection depend on the SNR regime the network operates in? 4) Can the delay be reduced by deployment of noisy feedback?

In this paper, we answer the above questions.
%%%%%%%%%%%%%%
%. 3. Our Goal
%%%%%%%%%%%%%%
%In this paper, our goal is to design multihop networks such that to minimize the latency in the network is minimized.
%Our goal is to minimize latency in multihop network by the optimum choice of a cooperative strategy at each relay and by the use of feedback. 
%%%%%%%%%%%%%%%%%%
% 4. Error Exponent
%%%%%%%%%%%%%%%%%
To characterize the delay introduced by coding, we derive error exponents achievable in a multihop network.
The error exponent characterizes the tradeoff between the code block size (and hence the delay) and the reliability \cite{GallagerBook}.
%For a point-to-point channel, the error exponent  is given by
%\begin{equation} \eqnlabel{E0}
%E_0 = - \lim_{n \rightarrow } \sup \frac{1}{n} P_e
%\end{equation}
Given the desired reliability and using the error exponent, we obtain the lower bound on the delay in the considered multihop network.
%%%%%%%%%%%%%%%%%%%%%%
% 5. Our Results
%%%%%%%%%%%%%%%%%%%%%%
We then formulate and solve an optimization problem that determines optimal selection of amplify-and-forward and decode-and-forward relays. We demonstrate that an approach in which the selection of AF and DF scheme at a relay is solely based on the received SNR is suboptimal.  We present the optimal DF/AF selections for several examples of networks. We prove that in high SNR the optimum transmission scheme is for all relays to perform amplify-and-forward. 
We show that in a symmetric network, all decode-and-forward nodes should be separated by an equal number of amplify-and-forward relays.

We then analyze  benefits of deploying noisy feedback. 
Our consideration of feedback is motivated by the well known fact that  feedback can improve the error exponents, in some scenarios drastically \cite{SchalkwijkKailath1966}.
The error exponents for the point-to-point channel with active noisy feedback for binary signaling were analyzed  in \cite{KimLapidothWeissman2011}. We extend some of these results to the multihop relay network. 
We first investigate the impact feedback has on the error exponent and thus delay in the single-relay channel. We then extend the analysis to the multihop network. 

%%%%%%%%%%%%%%%
{\it Related Work}

%%%%%%%%%%%%%%
 For discrete multihop networks with DF relays reliability bounds were analyzed in \cite{Oyman2006}.  Minimizing the delay in Gaussian multihop networks with DF relays was presented in \cite{WenBerry2008}. Error exponents in multihop network with AF relays were analyzed in \cite{NgoLarsson2011}. Minimizing latency over a microwave networks with DF and AF nodes by considering channel capacity was considered in \cite{Jorgensen2012}.

%%%%%%%%%%%%%%%%%%%%%
% Paper organization
%%%%%%%%%%%%%%%%%%%%%
The paper is organized as follows. In Section \ref{MultihopNetwork} we define and solve the optimization problem that determines the optimal selection of DF and AF relays.
Section \ref{Feedback} analyzes the impact of feedback. Section \ref{Conclusion} concludes the paper and discusses future work. In the paper, the proofs of theorems are outlined. Detailed proofs are available in \cite{MarictoSubmit}.

%%%%%%%%%%%%%%%%%%%%%%%%%%%
\section{Multihop Network} \label{MultihopNetwork}
%%%%%%%%%%%%%%%%%%%%%%%%%%%
% 1. Network and Channel Model
We consider a single source-destination multihop wireless network in which data from the source to the destination is transmitted via $H$ relays (see Fig.~\ref{Fig2}). 
%All nodes are full-duplex. 
Each node is equipped with a single antenna. Following practical constraints, we assume that all transmissions are orthogonal and that each node communicates only with its neighbor, as indicated  in Fig.~\ref{Fig2}.  We consider a Gaussian channel where a transmitted signal is corrupted by the additive, white Gaussian noise. The received signal at the node $k$ is  given by
\begin{align} \eqnlabel{channelmodel}
y_k = h_{k-1} x_{k-1} + z_k \qquad k=1, \ldots, H
\end{align}
where the transmit signal at the source  is denoted  $x_0$. The channel gain from a node $k-1$ to node $k$ is denoted $h_{k-1}$. Noise $z_k$ has zero mean and variance $\sigma^2_k.$ 
Similarly, the received signal at the destination is 
\begin{align} \eqnlabel{yD}
y_D = h_H x_H + z_D
\end{align}
where $z_D$ is zero mean with variance $\sigma^2_D$.
To simplify the presentation we assume in this section that $\sigma^2_k =\sigma^2$ for all $k=1, \ldots, H$ and $k=D$.
The power constraint at node $k$ is given by 
\begin{equation} \eqnlabel{Pk}
E[X_k^2] \le P_k.
\end{equation}
The source sends $B$ bits information intended for the destination node $D$ using a codeword of length $n_0$,  by sending a message $W$ from the message set $\Wc =\{1, \ldots, 2^B \}$. The encoding function at the source is given by $X_0^{n_0} = f(W)$. A general encoding function at each relay $k$  at time $i$ is given by $X_{k,i} = f_{k,i}(Y_k^{i-1})$. 
%%%%%%%%%%%
Each relay in the network performs either decode-and-forward or amplify-and-forward. 
The decode-and-forward scheme does not require block Markov encoding \cite{CoverElGamal79} because each receiving node receives signal only from one other node. Let $K \le H$ denote the number of DF relays. 
The $k$th DF relay performs decoding $\hat W_k = g_k(Y_k^{n_{k-1}})$ where $\hat W_k$ denotes the message estimate at that node and $n_{k-1}$ denotes codelength used by $(k-1)$th DF node. After decoding, the $k$th relay sends a codeword of length $n_k$: $X_k^{n_k} = f_k(\hat W_k)$. On the other hand, a relay $k$ performing AF, at each time instant $i$ transmits 
\begin{align}
x_k(i) = \beta_k y_k(i-1)  \eqnlabel{afnof}
\end{align}
where $\beta_k$ denotes the amplification gain. From \eqnref{afnof} and due to the power constraint at the relay \eqnref{Pk},  $\beta_k$ satisfies:
\begin{equation}
\beta_k^2 \le \frac{P_k}{h_{k-1}^2P_{k-1} + \sigma^2}. \eqnlabel{beta}
\end{equation} 
The decoding function at the destination is given by $\hat W = g(Y_D^{n_K})$ where $n_K$ denotes the codelength used at the $K$th DF relay in the network. The average error probability of the code is $P_e^{(n_K)} = P[W \neq \hat W(Y_D^{n_K})]$.

%% 2. Problem.
Our goal is to minimize the delay in sending messages between the source and the destination while guaranteeing a required level of reliability $\delta_e$ at the destination, i.e.,
\begin{equation}\eqnlabel{reliability}
P_e^{(n_K)} \le \delta_e.
\end{equation}
 We consider a problem in which the network is already in place, i.e, the relays are already positioned in the network. Therefore, the number of hops $H$ and channel gains are given. The considered multihop network is typically a microwave network with a high capacity line-of-sight channel at each hop. The channel variations are much slower compared to a cellular network and thus a transmitter typically has the channel state information. Our goal is to determine a cooperative strategy  such that the end-to-end delay is minimized  while guaranteeing a required level of reliability \eqnref{reliability}.
%We consider two possible cooperative strategies at the relays: decode-and-forward and amplify-and-forward.
In order to minimize the delay, we next review the error exponents and the delay associated with decode-and-forward and amplify-and-forward cooperative schemes.

% 3. Error Exponents with decode-and-forward and amplify-and-forward

The error exponent is defined by \cite{GallagerBook}
\begin{equation} \eqnlabel{E0}
E_r = - \lim_{n \rightarrow \infty} \sup \frac{1}{n} P_e^{*(n)}
\end{equation}
where $P_e^{*(n)}$ denotes the infimum of the error probability over all $(R,n)$ codes.
In a Gaussian point-to-point  channel with a received signal-to-noise ratio denoted as SNR, by choosing the Gaussian inputs the error exponent \eqnref{E0} evaluates to 
\begin{equation} \eqnlabel{E0Gauss}
E_r\ge \max_{\rho \in [0,1]} \left[ \rho \log(1 + \frac{\text{SNR}}{1+\rho}) - \rho R \right].
\end{equation}
In order to satisfy a reliability constraint \eqnref{reliability},
the delay introduced by transmission of a codeword of length $n_{pp}$  can be calculated from  \eqnref{E0} and \eqnref{E0Gauss} to be
\begin{equation} \eqnlabel{delay}
n_{pp} \ge \frac{\rho B-\log \delta_e}{\log(1+\frac{\text{SNR}}{1+\rho})}
\end{equation}
where $\rho \in [0,1]$ should be chosen so that $n_{pp}$ is minimized.

Consider a multihop network with $K$ relays all performing decode-and-forward. 
%Note that the scheme does not require block Markov encoding  \cite{CoverElGamal79} because each receiving nodes receives a signal only from one other node. 
%Each DF relay $k$ sends a codeword $X_k = f_k(\hat W_k)$ where $\hat W_k$ denotes the message estimate at node $k$.
Because at each hop the information is decoded,
each relay introduces a delay given by \eqnref{delay} and the total delay obtained from \eqnref{delay} is, \cite{WenBerry2008}
\begin{equation} \eqnlabel{delayDF}
D_{DF} \ge \sum_{k=1}^{K+1} \frac{\rho_k B-\log \delta_k}{\log(1+\frac{\text{SNR}_k}{1+\rho_k})}
\end{equation}
where $\delta_k$ denotes the required level of reliability at relay $k$, $\text{SNR}_k$ is received SNR at relay $k$ and we denoted the destination node as $K+1$.
%In the special case of a {\it symmetric} network with equal transmit powers, channel gains, noise variances and desired reliability levels, the delay reduces to 
%\begin{equation}
%n_{df} \ge (K+1) \frac{\rho L-\log \delta}{\log(1+\frac{\text{SNR}}{1+\rho})}.
%\end{equation}

%%%%%%%%%%%%%%%%%%%%%%%
% AF
%%%%%%%%%%%%%%%%%%%%%%%
We consider the amplify-and-forward cooperative strategy next.  
%The transmit signal at relay $k$ at time $i$ is given by
%\begin{align}
%x_k(i) = \beta_k y_k(i-1)  \eqnlabel{afnof}
%\end{align}
%where $\beta_k$ denotes the amplification gain. From \eqnref{afnof} and due to the power constraint at the relay \eqnref{Pk},  $\beta_k$ satisfies:
%\begin{equation}
%\beta_k^2 \le \frac{P_k}{h_{k-1}^2P_{k-1} + \sigma^2_k}. \eqnlabel{beta}
%\end{equation}
When all $K$ relays perform amplify-and-forward, it is straightforward to derive from \eqnref{channelmodel}, \eqnref{yD}, \eqnref{afnof} and \eqnref{beta}  that the received signal at the destination can be written as
\begin{equation} \eqnlabel{equivalentchannel}
y_D(i) = h_{e,K} x_0(i-K) + z_{e,K}(i) 
\end{equation}
where 
\begin{align}
%& h_e(K) = h_1 \prod_{i=2}^K \beta_i h_i \nonumber \\
%& z_e(K) = \sum_{k=2}^K \left( \prod_{i=k}^K \beta_i h_i \right) z_i 
& h_{e,K} = h_0 \prod_{i=1}^K \beta_i h_i \nonumber \\
& z_{e,K}(i) = \sum_{k=1}^K \left( \prod_{j=k}^K \beta_j h_j \right) z_k(i-K+k-1) + z_D(i). 
\end{align}
We denote the received SNR in the equivalent channel \eqnref{equivalentchannel} as $\gamma(K)$:
\begin{equation} \eqnlabel{gammaK}
%\gamma(K) = \text{SNR} \frac{(\prod_{i=2}^K \beta_i h_i)^2}{(\sum_{k=2}^K \prod_{i=k}^K \beta_i h_i )^2}.
\gamma(K) = \text{SNR}_0\frac{\prod_{i=1}^K (\beta_i h_i)^2}{\sum_{k=1}^K \prod_{i=k}^K (\beta_i h_i )^2 + 1}
\end{equation}
where $\text{SNR}_0 = h_0^2 P_0/\sigma^2$.
We observe that the output \eqnref{equivalentchannel} is the same as in a point-to-point channel with received SNR $\gamma(K)$. From \eqnref{delay}, it then follows that the delay introduced by $K$ AF relays is given by 
\begin{equation}\eqnlabel{delayAF}
D_{AF} \ge \frac{\rho B-\log \delta_e}{\rho \log(1+\frac{\gamma(K)}{1+\rho})}.
\end{equation}
where $\delta_e$ denotes required end-to-end reliability.

We observe from \eqnref{delayDF} that each DF relay introduces a delay $n_{pp}$ given by \eqnref{delay}. This is due to the fact that a DF node has to wait to receive the whole codeword prior to forwarding. In contrast, an AF node can forward on symbol-per-symbol, reducing a network of a cascade of AF nodes to a point-to-point channel \eqnref{equivalentchannel}, albeit with reduced SNR. Each AF node reduces the received SNR by amplifying the noise thereby reducing the transmission rate and ultimately increasing the delay. In fact, below a certain value of $\gamma(K)$, an AF node will cause a larger delay than a DF node. 
%We will illustrate this point later in the section (see Fig.~\ref{Fig3}).   

% 4. Optimization Problem
To determine the optimum number and positions of AF and DF relays that minimize the end-to-end delay in the considered multihop network, we next define a following optimization problem.
Let $N_{DF}$ denote the number of DF nodes in the network including the source.
Let $K_i \in \{0, \ldots, H \}$ denote the number of AF relays in between $(i-1)$th and $i$th DF relay and
let $p_i$ denote the index (position) of the $i$th DF relay. Then, $p_1=0.$
The delay introduced between  $(i-1)$th and $i$th DF relay is given by \eqnref{delayAF} for $K=K_i$.
We formulate the optimization problem as:
\begin{align} \eqnlabel{genoptimization}
&D^* = \min_{N_{DF},K_i} \sum_{i=1}^{N_{DF}} \frac{\rho_i B-\log \frac{\delta_e}{N_{DF}}}{\rho_i \log(1+\frac{\gamma(p_i,K_i)}{1+\rho_i})} \nonumber \\
&\text{s.t. } \sum_{i=1}^{N_{DF}} K_i + N_{DF} = H+1
\end{align}
where
\begin{equation} \eqnlabel{GammaKi}
\gamma(p_i,K_i) = \frac{\text{SNR}_{p_i} \prod_{j=p_i+1}^{p_i+K_i} (\beta_j h_j)^2}{\sum_{k=p_i+1}^{p_i+K_i} (\prod_{j=k}^{p_i+K_i} \beta_j h_j )^2  +1}
\end{equation}
and  $\text{SNR}_{p_i} = h_i^2 P_i/\sigma^2$.
The reliability constraint at each hop is chosen such that, by union bound, the end-to-end reliability constraint is satisfied.
The solutions to this problem can efficiently be found by dynamic programming \cite{BookDP}.
We present solution for examples of networks later in this section.
%We next present several properties of the optimal placements for the DF nodes.

%%%%%%%%%%%%%%%
% 1. High SNR
%%%%%%%%%%%%%%%
The next theorem present a solution to \eqnref{genoptimization} in high SNR. It shows that in the high SNR regime, the minimum delay is obtained when all relays perform AF.
\begin{MyTheorem} \thmlabel{HighSNR}
Let $P_k = s \hat P_k$ for all $k = 0, \dots, H$. Then
\begin{equation}
\lim _{s \rightarrow \infty} \frac{D_{AF}(s)}{D_{DF}(s)} = \frac {1}{N_{DF}}.
\end{equation}
\end{MyTheorem}
\begin{proof} ({\it Outline})
In high SNR, the SNR after $K$ stages of AF relays, \eqnref{gammaK}, reduces to
\begin{equation}
\gamma (K) = \left( \sum_{k=0}^K \frac{1}{\text{SNR}_k} \right)^{-1}
\end{equation}
where we used the notation $\text{SNR}_k = h_k^2P_k/\sigma^2$.
The rest of the proof follows from evaluating the delay associated with AF relays \eqnref{delayAF} and with DF relays \eqnref{delayDF}.
\end{proof}
\begin{MyRemark}
The intuition for this solution comes from the fact that AF introduces a smaller delay than DF for the price of reduced SNR for each AF hop.
In high SNR, however, the SNR loss is negligible and therefore multihop AF is optimal.
\end{MyRemark}

%%%%%%%%%%%%%%%%%%%%%
% 2. Symmetric network
%%%%%%%%%%%%%%%%%%%%%
We next consider a symmetric network with equal power and channel gains.
We consider a subproblem of \eqnref{genoptimization} wherein we optimize $K_i$, $i=1, \ldots, N_{DF}$ for fixed $N_{DF}$.
To give an insight to the optimum solution, we relax the constraint that $K_i$ is an integer.
We have the following result.
\begin{MyLemma}
For a symmetric network with $N_{DF}$ decode-and-forward nodes, the optimum solution satisfies for all $i = 1, \ldots, N_{DF}$
\begin{equation} \eqnlabel{symmetricsolution}
K_i^* = \frac{H}{N_{DF}}.
\end{equation}
\end{MyLemma}
\begin{proof} ({\it Outline})
The proof follows by forming the Lagrangian for the optimization problem \eqnref{genoptimization} and by deriving the optimality conditions.
%We denote the $i$th delay term in \eqnref{genoptimization} as $d_i(K_i)$ and form the Lagrangian for the optimization problem \eqnref{genoptimization}.
%From optimality conditions, we obtain
%\begin{equation}
%d'_j=\lambda
%\end{equation}
%that is satisfied for all $K_j$ being equal yielding \eqnref{symmetricsolution}.
%By substituting the constraint into the objective function in in \eqnref{genoptimization} , the objective function \eqnref{genoptimization} becomes
%\begin{equation}\eqnlabel{Dsum}
%D = \sum_{i=1}^{n_{DF}-1} d(K_i) + d \left( H - \sum_{i = 1}^{n_{DF} - 1} K_i \right)
%\end{equation}
%where 
%\begin{equation}\eqnlabel{Di}
%d_i = \frac{\rho L-\log \delta}{\rho \log(1+\frac{\gamma(K)}{1+\rho})}.
%\end{equation}
%From \eqnref{Dsum}, it is easy to show that the optimality conditions for all $j=1, \ldots, n_{DF}-1$ are satisfied when
%\begin{equation}
%K_j = H - \sum K_i 
%\end{equation}
%and therefore, all $K_j$ are equal, yielding \eqnref{symmetricsolution}.
\end{proof}
Therefore the optimal positions of AF and DF nodes are such that all AF nodes are positioned at the equal distance.

%%%%%%%%%%%%%%%%%%%%%%%%%%
% General network
%%%%%%%%%%%%%%%%%%%%%%%%%%%
%We next consider a general network with arbitrary transmit powers and channel gains. What allowed the above proof is the fact that all the terms in \eqnref{optimization} were the %same. This is no longer the case in an arbitrary network.  We fix an equal value for $\rho$ in \eqnref{optimization}. We next optimize over effective SNR values, that is we %reformulate the problem as
%\begin{align} \eqnlabel{optimization}
%&D^* = \min_{\gamma_i(K_i)} \sum_{i=1}^{n_{DF}} \frac{\rho L-\log \delta}{\rho \log(1+\frac{\gamma_i(K_i)}{1+\rho})} \nonumber \\
%&\text{s.t. } \sum_i K_i = H.
%\end{align}
%Similarly to \eqnref{Dsum} we have, 
%\begin{equation}
%D = \sum_{i=1}^{n_{DF}-1} d_{\gamma}(\gamma_i(K_i)) + d_{\gamma}(H - \sum_{i = 1}^{n_{DF} - 1} K_i)
%\end{equation}
%where 
%\begin{equation}\eqnlabel{Di}
%d_{\gamma}(\gamma_i(K)) = \frac{\rho L-\log \delta}{\rho \log(1+\frac{\gamma_i(K_i)}{1+\rho})}.
%\end{equation}
%Again, all terms are equal and the optimality conditions reduce to having all $\gamma(K)$ to be equal.

%%%%%%%%%%%%%%%%%%%
% 3. Examples.
%%%%%%%%%%%%%%%%%%
For a network with a single relay, Figure~\ref{Fig3} shows the delay associated with DF and AF, as a function of channel gains. Channel gains to and from the relay are chosen equal. We observe that, for smaller values of channel gains, DF at the relay is optimal. After a certain threshold, AF becomes optimal. 
\begin{figure}[t]
\centering
\includegraphics[height=6cm, width=3.5in]{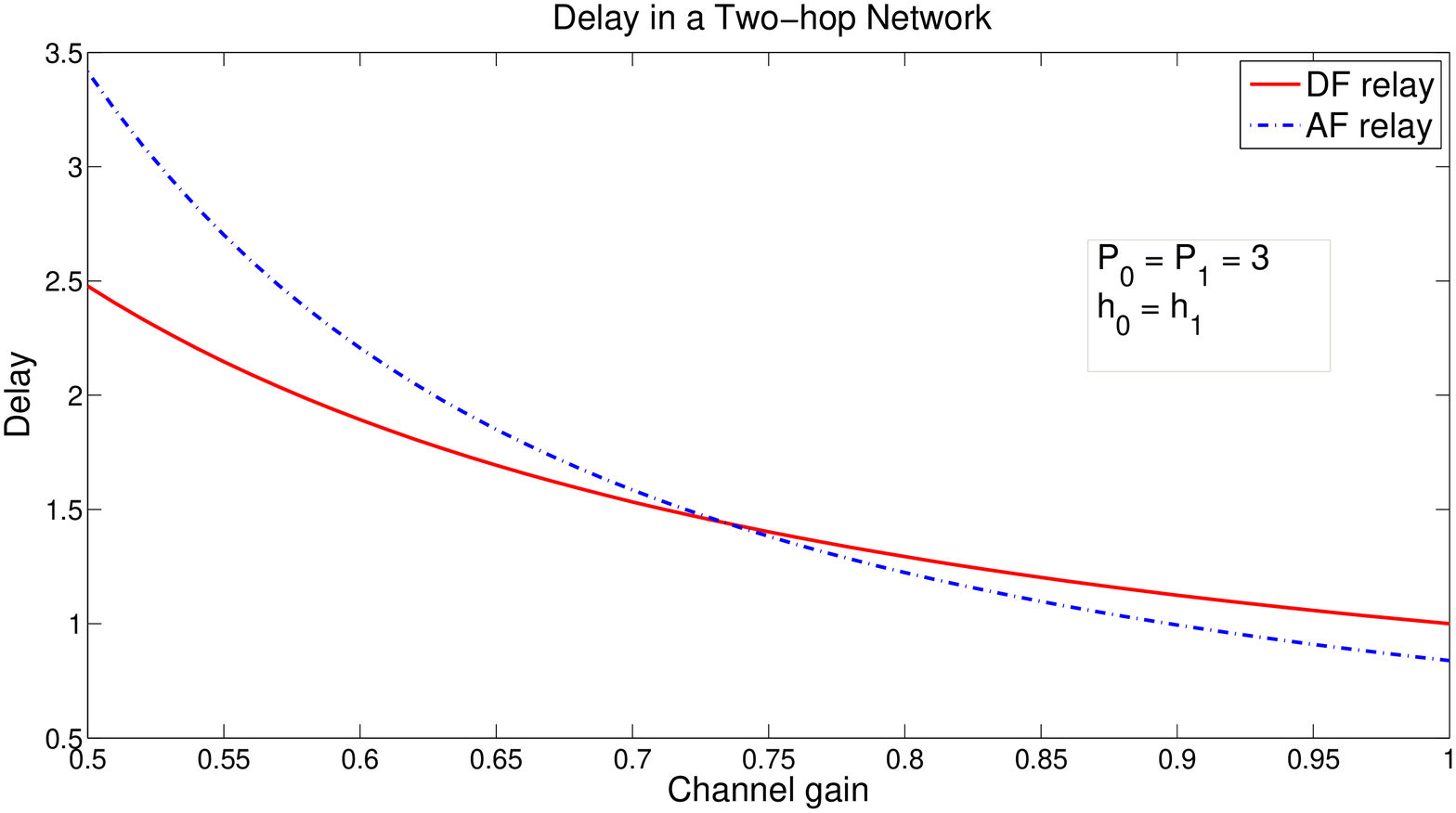}
\caption{Delay in a two-hop network.}
      \label{Fig3}
\end{figure}

\begin{figure}[t]
\centering
\includegraphics[height=.9cm, width=3.5in]{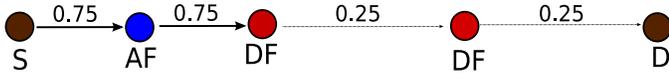}
\caption{Optimal selection of AF and DF relays in a 4-hop network. Transmit powers are chosen equal at all nodes. Channel gains are shown for each hop. }
      \label{Fig3E}
\end{figure}

\begin{figure}[t]
\centering
\includegraphics[height=.9cm, width=3.5in]{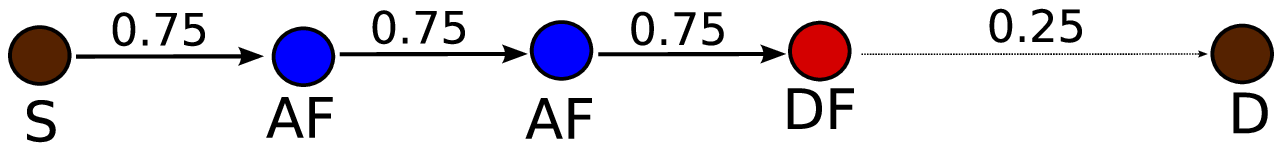}
\caption{Optimal selection of AF and DF relays  in a 4-hop network. Channel gains are shown for each hop.}
      \label{Fig4E}
\end{figure}
\begin{figure}[t]
\centering
\includegraphics[height=.9cm, width=3.5in]{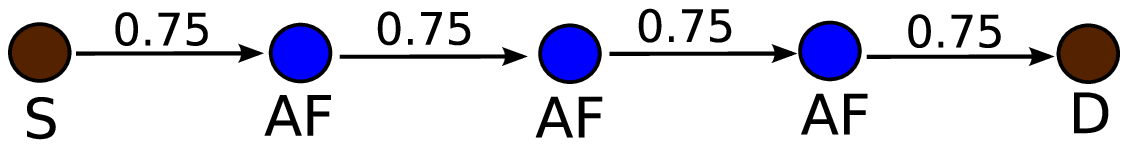}
\caption{Optimal selection of AF and DF relays in a 4-hop network. Channel gains are shown for each hop. This network is in high SNR.}
      \label{Fig5E}
\end{figure}

%\begin{figure}[t]
%\centering
%\includegraphics[height=5cm, width=3.4in]{3HopExample}
%\caption{Optimal positions of AF and DF nodes for different channel gain values in a 3-hop network. Transmit powers are chosen equal at all nodes.}
%      \label{Fig4}
%\end{figure}
%Figure~\ref{Fig4} shows a solution of optimization problem \eqnref{genoptimization} in a 3-hop network, i.e., the optimal positions of AF and nodes for $4$ sets of channel gain %values. We observe that, if the received SNR (i.e. channel gain) is small, the relay performs DF. However, observe that in case $2$, the second relay performs DF despite the fact %that its received channel gain from  Relay 1 is large. This is due to the fact that the channel gain on the next hop is small. This demonstrate that an approach in which the %selection of AF and DF nodes is done solely based on the received SNR (as is a common practice) is suboptimal. Case 4 is the case of high SNR regime in which it is optimal for %both relays to perform AF.

Figures~\ref{Fig3E}-\ref{Fig5E} show a solution of optimization problem \eqnref{genoptimization} for three examples of 4-hop networks. In all three examples the transmit powers at nodes are chosen equal.
We observe that the optimal relay selection in Fig.~\ref{Fig3E}  differs from the common practice solution whereby only relay $3$ would perform DF. We observe that in the optimum solution, relay $2$ chooses to perform DF to compensate for weak links further down in the transmission chain. In contrast, if relay $2$ would make decision solely based on its received SNR, it would choose to perform AF.     Another example in which the optimal selection differs from the common practice is shown in \ref{Fig4E}. Fig.~\ref{Fig5E} shows that for large values of channel gains the optimal solution is AF for all relays, in agreement with Theorem \thmref{HighSNR}. 

%%%%%%%%%%%%%%%%%%%%%%%%%%%%%%%%%%%%%%%%%%%%%%%%%%%%%%%%%%%%%%%%%%%
\section{Delay in Networks with Feedback} \label{Feedback}
%%%%%%%%%%%%%%%%%%%%%%%%%%%%%%%%%%%%%%%%%%%%%%%%%%%%%%%%%%%%%%%%%%%%
%%%%%%%%%%%%%%%%%%%%%%%%%%%%%%%%%%%%%%%%%%%%
\begin{figure}[t]
\centering
\includegraphics[height=3.5cm, width=3.5in]{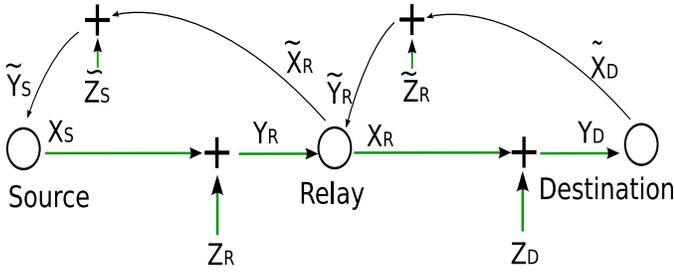}
\caption{Single relay channel with per-hop feedback. No direct link is assumed between the source and the destination.}
      \label{Fig1}
\end{figure}
%%%%%%%%%%%%%%%%%%%%
% Channel model
%%%%%%%%%%%%%%%%%%%%
To examine the impact of feedback, we first analyze the network with a single relay. 
We assume that from each receiving node i.e., the relay and the destination,  there is a feedback link (see Fig.~\ref{Fig1}). The feedback links are orthogonal to the forward links.
The forward channel output at the relay is given by \eqnref{channelmodel} and at the destination by \eqnref{yD}, for $H=1$. Since there is only one relay, we denote the source and relay inputs as $x_S$ and $x_R$ respectively. We use similar notation also for the outputs and the noise at each receiver. The input-output relationship is then given by:
\begin{align} \eqnlabel{FeedbackChannelModel}
&y_R = h_S x_S + z_R \nonumber \\
&y_D = h_R x_R + z_D.
\end{align}
%where $x_S$ and $x_D$ denote respective inputs at the source and the relay, $y_R,y_D$ denote respective outputs at the relay and the destination, and $z_R$ and $z_D$ represent %white Gaussian noise with respective variances $\sigma_R$ and $\sigma_D$. 
As before, the source and the relay satisfy respective power constraints $E[X_S^2] \le P_S$ and $E[X_R^2] \le P_R$.

Similarly to the forward channel, the feedback channel is:
\begin{align} \eqnlabel{ChannelModelFeedback}
&\tilde y_R = \tilde h_R \tilde x_D + \tilde z_R \nonumber \\
&\tilde y_S = \tilde h_S \tilde x_R + \tilde z_S 
\end{align}
and the power constraints are given by  $E[\tilde X_R^2] \le \tilde P_R$ and $E[\tilde X_D^2] \le \tilde P_D$. Noises $\tilde z_R$ and $\tilde z_S$ have zero mean and respective variance $\tilde \sigma^2_R$ and $\tilde \sigma^2_S$.

To analyze the impact of feedback on the delay in this channel, we extend the results of \cite{KimLapidothWeissman2011} that develops error exponents for the point-to-point channel with active noisy feedback.  The analysis in \cite{KimLapidothWeissman2011} assumes binary signaling, and in the reminder of the paper, we make the same assumption. In particular, we assume that the encoder sends a single bit that takes values $0$ or $1$ equiprobably. 
We have the following theorem.
\begin{MyTheorem}\thmlabel{Feedback1}
The error exponent achievable in the considered single-relay channel with active noisy feedback is bounded by
\begin{equation}\eqnlabel{EFB}
E_{FB} \ge \frac{2P_S}{\sigma_F^2} +\frac{2 \tilde P_{D}}{\sigma_{FB}^2}
\end{equation}
where
\begin{align} \eqnlabel{sigmaF}
&\sigma^2_F = \frac{\sigma^2_R}{h_S^2} + \frac{\sigma^2_D}{(h_Sh_R \beta)^2}  \\
&\sigma^2_{FB} = \frac{\tilde \sigma^2_R}{\tilde h_R^2} + \frac{\tilde \sigma^2_S}{(\tilde h_S \tilde h_R \tilde \beta)^2} \nonumber
\end{align}
and  the corresponding delay for the reliability level $\delta_e$ is bounded by
\begin{equation}\eqnlabel{Delay}
 n_{FB} \ge \left( \frac{2P}{\sigma_F^2} + \frac{2 \tilde P}{\sigma_{FB}^2}  \right)^{-1}  \log \frac{1}{\delta_e}.
\end{equation}
\end{MyTheorem}
\begin{proof}
We consider the amplify-and-forward cooperative strategy at the relay both on the forward and the backward channel.  
The transmit signal at the relay at time $i$ in the forward channel is then given by \eqnref{afnof} for $k=R$
\begin{align}
x_R(i) = \beta y_R(i-1)  \eqnlabel{af}
\end{align}
where $\beta$ from \eqnref{beta} satisfies:
\begin{equation}
\beta^2 \le \frac{P_R}{h_s^2P_S + \sigma^2_R}. \eqnlabel{betafeedback}
\end{equation}
The received signal at the destination evaluates from \eqnref{equivalentchannel} to be
\begin{align} \eqnlabel{EquivalentChannel}
%y_D(i) & = h_2 \beta y_R(i-1) + z_D(i) \nonumber \\
%& = h_2 \beta (h_1 x_S(i-1) + z_R(i-1)) + z_D(i) \nonumber \\
%& = h_2 \beta h_1 x_S(i-1) + h_2 \beta z_R(i-1) + z_D(i) \nonumber \\
y_D(i)& = h_{eq}x_S(i-1) + z_{eq}(i)
\end{align}
where  
\begin{equation} \eqnlabel{heq}
h_{eq} = h_R \beta h_S
\end{equation}
and
\begin{equation} \eqnlabel{zeq}
z_{eq}(i) = h_R \beta z_R(i-1) + z_D(i). 
\end{equation}
Channel \eqnref{EquivalentChannel} is equivalent to the unit-gain channel given by
\begin{align} \eqnlabel{forwardchannelunitvariance}
y_D(i)& = x_S(i-1) + z_F(i)
\end{align}
where $z_F = z_{eq}/ h_{eq}$. Using \eqnref{heq} and \eqnref{zeq} we obtain that the variance of $z_F$   equals $\sigma^2_F$ given by \eqnref{sigmaF}.

Similarly, the transmit signal at the relay in the reverse channel is given by
\begin{align}
\tilde x_R(i) =\tilde \beta \tilde y_R(i-1)  \eqnlabel{afF}
\end{align}
with amplification gain
\begin{equation}
\tilde \beta^2 \le \frac{\tilde P_R}{\tilde h_R^2 \tilde P_D + \tilde \sigma^2_R}. \eqnlabel{betafeedback}
\end{equation}
The received feedback signal at the source is given by
\begin{align} \eqnlabel{EquivalentChannelSource}
\tilde y_S(i) & = \tilde h_S \tilde \beta \tilde h_R \tilde x_D(i-1) + \tilde h_S \tilde \beta \tilde z_R(i-1) + \tilde z_S(i) \nonumber \\
& = \tilde h_{eq}x_D(i-1) + \tilde z_{eq}(i)
\end{align}
where
\begin{equation}
\tilde h_{eq} = \tilde h_S \tilde \beta \tilde h_R, 
\end{equation}
and
\begin{equation}
\tilde z_{eq}(i) = \tilde h_S \tilde \beta \tilde z_R(i-1) + \tilde z_S(i). 
\end{equation}
Again, we can consider the equivalent unit-gain channel
\begin{align} \eqnlabel{feedbackchannelunitvariance}
\tilde y_S(i) & = x_D(i-1) + \tilde z_{FB}(i)
\end{align}
where $z_{FB} = \tilde z_{eq}/\tilde h_{eq}$ has the variance $\sigma^2_{FB}$ given by \eqnref{sigmaF}.

The equivalent channel model given by \eqnref{forwardchannelunitvariance} and \eqnref{feedbackchannelunitvariance}  corresponds to a Gaussian point-to-point channel with a noisy active feedback with noise variances given by \eqnref{sigmaF}, analyzed in \cite[Sec. VII]{KimLapidothWeissman2011}.
 The result in \cite[Sec. VII]{KimLapidothWeissman2011}  applies yielding the error exponent given by \eqnref{EFB}.
\end{proof}
\begin{MyRemark}
We compare the delay \eqnref{Delay} to the delay in the point-to-point channel with the respective noise variances in the forward and the feedback channel $\sigma^2_{D,pp}$ and $\tilde \sigma^2_{S,pp}$, and the respective channel gains $h_{S,pp}$ and $\tilde h_{D,pp}$. The error exponent of this channel, as determined in \cite[Sec. VII]{KimLapidothWeissman2011}, is given by \eqnref{EFB} with $\sigma^2_{F} = \sigma^2_{D,pp}/h^2_{S,pp}$ and $\sigma^2_{FB}=\tilde \sigma^2_{S,pp}/\tilde h^2_{D,pp}$. By comparing \eqnref{EFB} for the relay with  the point-to-point channel, we obtain the sufficient conditions under which the delay in the relay channel is smaller: 
\begin{align}\eqnlabel{conditionRemark2}
&\frac{\sigma^2_R}{h_S^2} + \frac{\sigma^2_D}{(h_Sh_R\beta)^2}< \frac{\sigma^2_{D,pp}}{h^2_{S,pp}} \\
&\frac{\tilde \sigma^2_R}{\tilde h_R^2} + \frac{\tilde \sigma^2_S}{(\tilde h_S \tilde h_R\tilde \beta)^2}< \frac{\tilde \sigma^2_{Spp}}{\tilde h^2_{D,pp}} \nonumber.
\end{align} 
Consider the case when all the noise variances and powers are the same and the channel is in high SNR. Then $\beta=1/h^2_S$, $\tilde \beta=1/\tilde h^2_R$ and conditions \eqnref{conditionRemark2} reduce to
\begin{align}
&\frac{1}{h_S^2} + \frac{1}{h_R^2}< \frac{1}{h^2_{S,pp}} \\
&\frac{1}{\tilde h_R^2} + \frac{1}{\tilde h_S^2}< \frac{1}{\tilde h^2_{D,pp}} \nonumber.
\end{align} 
When the above conditions on the channel gains are satisfied, the error exponent and the delay are improved by the help of the relay.
\end{MyRemark}
\begin{MyRemark}
We compare the delay \eqnref{Delay} to the delay in the considered relay channel when there is no feedback.
As before, the relay performs amplify-and-forward and the channel is given by \eqnref{EquivalentChannel}. 
The error probability with binary signaling for this channel is given by 
\begin{equation}
P_b = \exp (-\frac{h^2_{eq}P_S}{2 \sigma_{eq}^2}n)
\end{equation}
where $h_{eq}$ is given by \eqnref{heq} and $\sigma_{eq}^2$ is the variance of $z_{eq}$ given by \eqnref{zeq}.
The corresponding delay for the reliability level $\delta_e$ is 
\begin{equation} \eqnlabel{DelayNoFeedback}
 n \ge \frac{2 \sigma_{eq}^2}{h_{eq}^2P_S} \log \frac{1}{\delta_e}.
\end{equation}
Delay \eqnref{Delay} and \eqnref{DelayNoFeedback} can be easily compared for case that $P=\tilde P$ and all noises have the same variance and $\sigma_F \ge \tilde \sigma_{FB}$ in \eqnref{sigmaF}. In this case, the delay with feedback is always smaller than the delay in the no-feedback channel.
\end{MyRemark}
\vspace{0.2cm}
%%%%%%%%%%%%%%%%%%%%%%%%%%%%%%%%%
% Multihop network with feedback
%%%%%%%%%%%%%%%%%%%%%%%%%%%%%%%%%
We can now extend the result of Thm.~\thmref{Feedback1} to multihop network with active noisy feedback between every transmitter/receiver pair as shown in Fig.~\ref{Multihop}.
 
\begin{MyTheorem}\thmlabel{Feedbackmultihop}
The error exponent achievable in the multihop relay network with $K$ relays and with active noisy feedback is bounded by
\begin{equation}\eqnlabel{EFBmultihop}
E_{FB} \ge \frac{2P_S}{\sigma_F^2} +\frac{2 P_{FB}}{\sigma_{FB}^2}
\end{equation}
where
\begin{align} \eqnlabel{sigmamultihop}
& \sigma^2_F =  \frac{\sum_{k=1}^H \prod_{i=k}^H (\beta_i h_i )^2 \sigma^2_k + \sigma^2_D}{h^2_0\prod_{i=k}^H (\beta_i h_i )^2 } \\
&\sigma^2_{FB} =  \frac{\sum_{k=1}^H \prod_{j=k}^H (\tilde \beta_j \tilde h_{j-1} )^2 \tilde \sigma^2_k + \tilde  \sigma^2_S}{\tilde h_0^2\prod_{k=1}^H (\tilde \beta_k \tilde h_k )^2 }. \nonumber
\end{align}
%and  the corresponding delay for the reliability level $\delta$ is bounded by
%\begin{equation}\eqnlabel{Delaymultihop}
%D \ge n \ge \left( \frac{2P}{\sigma_D^2} + \frac{2 \tilde P}{\sigma_{FB}^2}  \right)^{-1}  \log \frac{1}{\delta}.
%\end{equation}
\end{MyTheorem}
\begin{proof}
The proof follows the same steps as the proof of Theorem~\thmref{Feedback1}.
\end{proof}
%The results of Theorem~\thmref{Feedbackmultihop} can be used to perform the optimization of the AF and DF relay positions as was done in the case of no-feedback in \eqnref{genoptimization}.
%%%%%%%%%%%%%%%%%%%%%%%%%%%%%%%%%%%%%%%%%%%%
\begin{figure}[t]
\centering
\includegraphics[height=1.5cm, width=3.5in]{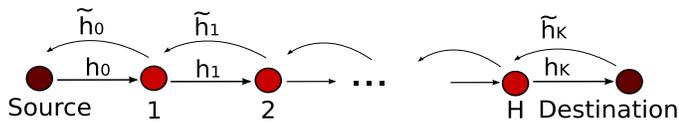}
\caption{Multihop network with feedback.}
      \label{Multihop}
\end{figure}
%%%%%%%%%%%%%%%%%%%%%%%%%%%%%%%%%%%%%%%%%%%%
\vspace{0.25cm}
{\it Discussion}

\vspace{0.25cm}
The above results show the improvement in terms of the required codelength $n$ over the no-feedback case. However,  the analysis presented in the previous section does not capture the propagation delay. In high-frequency trading applications where data is sent over multiple hops and hundreds of kilometers, the propagation delay is a dominant factor. Using feedback at every link is not appropriate as it would significantly increase the propagation delay. Instead, the presented analysis points that a feedback could potentially reduce the delay when used  sporadically, on certain hops.
%%%%%%%%%%%%%%%%%%%%%%%%%%%%%%%%%%%%%%
\section{Summary and Future Work} \label{Conclusion}
%%%%%%%%%%%%%%%%%%%%%%%%%%%%%%%%%%%%%%
We have presented an information-theoretic approach that, based on error exponents,  optimally determines a cooperative strategy between DF and AF for every relay.  We showed that in high SNR the optimum transmission scheme is for all relays to perform AF. We have then derived the error exponents in the considered multihop network in the presence of noisy, active feedback.

Our results demonstrate that the choice of the cooperative scheme cannot be made solely based on the received SNR at that relay, as is a common practice. Instead, this choice depends on the channel conditions in the whole network.

In practice, in addition to the average power constraint, the peak power constraint needs to be satisfied. This motivates extending the presented analysis under the peak power constraints. Achievable error exponents under the peak power constraints for the point-to-point channel were derived in \cite{XiangKim2010}.
Furthermore, the considered problem can be extended to include slow fading analysis. 
In addition to decode-and-forward and amplify-and-forward schemes analyzed in this paper, these results can be extended also to compress-and-forward \cite{CoverElGamal79}. 
Another open problem is investigating whether a form of feedback, possibly used only at certain time instants and at certain nodes, could improve the end-to-end delay.
%Finally, an important open problem is determining the lower bounds on the error exponent for the relay networks.
%%%%%%%%%%%%%%%%%%%%%%%%%%%%%%%%%%%%%%
\section{Acknowledgement}
%%%%%%%%%%%%%%%%%%%%%%%%%%%%%%%%%%%%%%
The author would like to thank Andrea Montanari for inspiring and insightful discussions.
%%%%%%%%%%%%%%%%%%%%%%%%%%%%%%%%%%%%%%%%%%%%
% Generated by IEEEtran.bst, version: 1.13 (2008/09/30)

%%%%%%%%%%%%%%%%%%%%%%%%%%%%%%%%%%%%%%%%%%%
\bibliographystyle{IEEEtran}
%\bibliography{referencesCOOP}
%\bibliography{/Users/ivanamaric/papers2011/DMTRELAY/referencesCOOP}

\begin{thebibliography}{10}
\providecommand{\url}[1]{#1}
\csname url@samestyle\endcsname
\providecommand{\newblock}{\relax}
\providecommand{\bibinfo}[2]{#2}
\providecommand{\BIBentrySTDinterwordspacing}{\spaceskip=0pt\relax}
\providecommand{\BIBentryALTinterwordstretchfactor}{4}
\providecommand{\BIBentryALTinterwordspacing}{\spaceskip=\fontdimen2\font plus
\BIBentryALTinterwordstretchfactor\fontdimen3\font minus
  \fontdimen4\font\relax}
\providecommand{\BIBforeignlanguage}[2]{{%
\expandafter\ifx\csname l@#1\endcsname\relax
\typeout{** WARNING: IEEEtran.bst: No hyphenation pattern has been}%
\typeout{** loaded for the language `#1'. Using the pattern for}%
\typeout{** the default language instead.}%
\else
\language=\csname l@#1\endcsname
\fi
#2}}
\providecommand{\BIBdecl}{\relax}
\BIBdecl

\bibitem{news}
R.~Martin, ``Wall streetÕs quest to process data at the speed of light,'' in
  \emph{Information Week,
  http://www.informationweek.com/wall-streets-quest-to-process-data-at-th/199200297},
  Apr. 2007.

\bibitem{PressRelease}
\emph{Press Release}.\hskip 1em plus 0.5em minus 0.4em\relax
  http://www.mckay-brothers.com/press/, 2013.

\bibitem{GallagerBook}
R.~Gallager, \emph{Information Theory and Reliable Communication}.\hskip 1em
  plus 0.5em minus 0.4em\relax Wiley, 1968.

\bibitem{SchalkwijkKailath1966}
J.~P.~M. Schalkwijk and T.~Kailath, ``A coding scheme for additive noise
  channels with feedback - part {I}: no bandwidth constraint,'' \emph{IEEE
  Trans. Inf. Theory}, vol.~12, no.~2, pp. 172--182, Apr. 1966.

\bibitem{KimLapidothWeissman2011}
Y.-H. Kim, A.~Lapidoth, and T.~Weissman, ``Error exponents for the {G}aussian
  channel with active noisy feedback,'' \emph{IEEE Trans. Inf. Theory},
  vol.~57, no.~3, pp. 1223--1236, Mar. 2011.

\bibitem{Oyman2006}
O.~Oyman, ``Reliability bounds for delay-constrained multi-hop networks,'' in
  \emph{Allerton Conference on Communications, Control and Computing,
  Monticello, IL}, Sep. 2006.

\bibitem{WenBerry2008}
N.~Wen and R.~A. Berry, ``Reliability constrained packet-sizing for linear
  multi-hop wireless networks,'' in \emph{IEEE Int. Symp. Inf. Theory}, Jul.
  2008.

\bibitem{NgoLarsson2011}
H.~Q. Ngo and E.~G. Larsson, ``Linear multihop amplify-and-forward relay
  channels: Error exponent and optimal number of hops,'' \emph{IEEE Trans. on
  Wireless Communications}, vol.~10, no.~11, pp. 3834--3842, nov 2011.

\bibitem{Jorgensen2012}
D.~Jorgensen, ``Minimising latency over microwave repeater network,'' in
  \emph{Aviat internal report}, Jun. 2012.

\bibitem{MarictoSubmit}
I.~Mari{\'c}, ``Low latency communications,'' \emph{Journal version, to be
  submitted}, 2013.

\bibitem{CoverElGamal79}
T.~Cover and A.~E. Gamal, ``Capacity theorems for the relay channel,''
  \emph{IEEE Trans. Inf. Theory}, vol.~25, no.~5, pp. 572--584, Sep. 1979.

\bibitem{BookDP}
S.~Dasgupta, C.~Papadimitriou, and U.~Vazirani, \emph{Algorithms}.\hskip 1em
  plus 0.5em minus 0.4em\relax
  http://www.cs.berkeley.edu/~vazirani/algorithms.html, 2006.

\bibitem{XiangKim2010}
Y.~Xiang and Y.-H. Kim, ``On the {AWGN} channel with noisy feedback and peak
  energy constraint,'' in \emph{Proc. IEEE Int. Symp. Inf. Theory}, Jun. 2010,
  pp. 256--259.

\end{thebibliography}
%%%%%%%%%%%%%%%%%%%%%%%%%%%%%%%%%%%%%%%%%%%%
\end{document}